\documentclass[12pt]{scrartcl}
\usepackage{cmap}
\usepackage[T1]{fontenc}
\usepackage[utf8]{inputenc}
\usepackage{lmodern}
\usepackage{mathtools,amssymb}
\usepackage{enumitem}
\usepackage[a4paper]{geometry}
\usepackage{amsthm}
\usepackage{csquotes}
\usepackage{interval}
\usepackage[backend=bibtex,sorting=nyt,style=authoryear,isbn=false,natbib=true]{biblatex}
\usepackage{faktor}
\usepackage[svgnames]{xcolor}
\usepackage{booktabs}
\usepackage{siunitx}
\usepackage{graphics,subfigure}
\usepackage{textcomp} 
\usepackage{marvosym}
\usepackage{microtype}
\usepackage[unicode,hyperfootnotes=false,pdftex]{hyperref}
\usepackage{hyperxmp}
\usepackage{bookmark}
\usepackage{upgreek}

\hypersetup{%
        pdftitle={A threshold model for local volatility: evidence of leverage and mean-reversion effects on historical data},
        pdfauthor={Antoine Lejay; Paolo Pigato},
        pdfdate={2019-02-15},
	pdflang={en},
	pdfkeywords= {Leverage effect; realized volatility; mean-reversion; regime-switch; parametric estimation; threshold diffusion; stock price model}
}

\newtheorem{proposition}{Proposition}

\theoremstyle{remark}

\newtheorem{remark}{Remark}

\pgfkeys{
  /interval/.is family,%
  /interval,%
    int left fence/.initial={\lbracket},
    int right fence/.initial={\rbracket},
    int/.style={left fence=\llbracket,right fence=\rrbracket}
}

\DeclareMathAlphabet{\Obs}{U}{}{}{}
\SetMathAlphabet\Obs{normal}{U}{eur}{m}{n}

\intervalconfig{soft open fences}

\newcommand\given{\nonscript\:\delimsize\vert\nonscript\:\mathopen{}} 
\newcommand\SetSymbol[1][]{\nonscript\:#1\vert\nonscript\:\mathopen{}\allowbreak}
\DeclarePairedDelimiterX\Set[1]\{\}{\renewcommand\given{\SetSymbol[\delimsize]}#1}
\DeclarePairedDelimiterX\Prb[1](){\renewcommand\given{\SetSymbol[\delimsize]}#1}

\DeclarePairedDelimiterXPP{\cro}[1]{}{\langle}{\rangle}{}{#1}
\DeclarePairedDelimiter{\abs}{|}{|}

\newcommand{\RR}{\mathbb{R}}

\newcommand{\PP}{\mathbb{P}}

\newcommand{\cN}{\mathcal{N}}

\newcommand{\ind}[1]{\mathbf{1}_{#1}}

\newcommand{\vd}{\,\mathrm{d}}

\DeclareMathOperator{\argmin}{arg\,min}

\DeclareMathOperator{\LogLik}{Log-Lik}

\newcommand{\eqdef}{\mathbin{:=}}

\DeclareSIUnit\year{yr}

{\begin{itemize}[leftmargin=1cm,noitemsep,label={$\star$}]\color{blue}}{\end{itemize}}

\begin{document}

\newgeometry{margin=2cm}

\title{
    A threshold model for local volatility:\\ evidence of leverage and mean-reversion \\ effects on historical data
}
    
\author{Antoine Lejay%
    \footnote{%
    Université de Lorraine, IECL, UMR 7502, Vand\oe uvre-lès-Nancy, F-54600, France\newline
    CNRS, IECL, UMR 7502, Vand\oe uvre-lès-Nancy, F-54600, France\newline
    Inria, Villers-lès-Nancy, F-54600, France\newline
    E-mail: \texttt{Antoine.Lejay@univ-lorraine.fr}}
    \ 
and Paolo Pigato%
    \footnote{%
    Weierstrass Institute for Applied Analysis and Stochastics, Mohrenstrasse 39, Berlin, 10117, Germany\newline
E-mail: \texttt{paolo.pigato@wias-berlin.de}
    }
}

\date{February 15, 2019}

\maketitle

\begin{abstract}
In financial markets, low prices are generally
associated with high volatilities and vice-versa,
this well known stylized fact  usually being referred to as
the leverage effect.
We propose a local volatility model, given by a stochastic
differential equation with piecewise constant coefficients, which accounts for
leverage and mean-reversion effects in the dynamics of
the prices. This model exhibits a regime switch in the dynamics accordingly to
a certain threshold. It can be seen as a continuous-time version of the self-exciting
threshold autoregressive (SETAR) model. We propose an estimation procedure for
the volatility and drift coefficients as well as for the threshold level.
Parameters estimated on the daily prices of 351 stocks of NYSE and S\&P 500, 
on different time windows, show consistent empirical evidence for leverage effects.
Mean-reversion effects are also detected, most markedly in crisis periods. 
\end{abstract}

\bigskip

\begin{quote}
\end{quote}

\restoregeometry

\section{Introduction}

Despite the predominance of the Black-Scholes model for the dynamics of
asset prices, its deficiencies to reflect all the phenomena observed in the
markets are well documented and subject to many studies.  Some \emph{stylized
facts} not consistent with the Black-Scholes model are  non-normality of
log-returns, asymmetry, heavy tails, varying conditional volatilities and 
volatility clustering \parencite{Cont:2001gv}.  Regime switching is also
consistently observed \parencite{ang,salhi}.  Besides, some assets and indices
exhibit mean-reverting effects 
\parencite[see \textit{e.g.}][]{meng13a,Monoyios:2002dk,Lo:1988hz,Poterba,SPIERDIJK2012228}.  

By considering only the asset's price at discrete times $\Set{k\Delta
t}_{k=0,1,2,\dotsc}$, the log-returns  $r_t=\log(S_{t+1}/S_t)$ of the Black-Scholes model~$\Set{S_t}_{t\geq 0}$ are nothing more than the simple time
series 
\begin{equation}
    \label{eq:discretebs:1}
    r_{t+1}=\left(\mu-\frac{\sigma^2}{2}\right)\Delta t
    +\sigma\sqrt{\Delta t}\epsilon_t\text{ with }\epsilon_t\sim\cN(0,1), \text{ independent}.
\end{equation}

Several models alternative to \eqref{eq:discretebs:1} have been proposed to take
some of these stylized facts into account. 
Among the most popular ones, ARCH and GARCH models and their numerous variants
reproduce volatility clustering effects~\parencite{engle}. 

In this article, we focus on \emph{leverage effects}, a term which refers to a
negative  correlation between the prices and the volatility.  As observed for a
long time, the lower the price, the higher the volatility.  First explanations
were given in~\citet{black1976,Christie:1982vf}.  Processes such as the
\emph{constant elasticity volatility}~(CEV) were proposed to account of these
phenomena~\parencite{Christie:1982vf}.  
One common economic
explanation of leverage effects is that when an asset price decreases, the ratio of the company's debt with
respect to the equity value becomes larger, and as a consequence volatility
increases; another explanation is that investors tend to become more nervous
after a large negative return than after a large positive return; anyway, the
origin of leverage effects is still subject to discussion
\parencite[see \textit{e.g.}][]{Hens:2006tv}. 

In the early 1980s, H.~Tong has proposed a broad class of time series, 
the \emph{threshold autoregressive models} (TAR), 
with non-linear effects reproducing cyclical data
\parencite{tong1983,Tong:2011ud,tong2015}. 
This class, which contains \emph{hidden Markov chains} (HMM) 
as well as \emph{self-exciting threshold autoregressive} models
(SETAR), produces a wide range of behaviors. 
HMM models rely on a temporal segmentation (they are good for crisis
detection), while SETAR models rely on a spatial segmentation, with a regime change when
the price goes below or above a threshold. 

Time series of SETAR type capture leverage and mean-reverting effects by
defining a threshold which separates two regimes (high/low volatility,
positive/negative trend).  Unlike models such as HMM, no external nor latent
randomness is used.

In finance, various aspects of SETAR like models have been considered 
\parencite{Yadav:1994km,Chen:2011bk,meng13a,Siu2016,Rabemananjara:1993dh}.
An alternative form to SETAR models is provided by \emph{threshold stochastic volatility}
models~\parencite{xu,So:2002gn,chen_liu_so2008}, where the volatility depends non-linearly on the price 
through a threshold model. Thresholds can also depend on auxiliary variables, as in \parencite{chen_so2006}. 
Also other considerations such as
\emph{psychological barriers} \parencite{FUT:FUT21648,kolb} lead to threshold
models.

Continuous-time models could be seen as the limit of time series as the time
step goes to $0$. They have some advantages over time series, for allowing
irregularly sampling, the use of stochastic calculus tools and possibly analytic or
semi-analytic formulas for fast evaluation of option
prices and risk estimation. Continuous-time threshold models (or \emph{threshold diffusion})
have been studied in~\citet{Siu2016,su2016} for option valuation, in~\citet{meng13a} for portfolio
optimization,~etc. Self-exciting variants 
of Vasi\v cek and Cox-Ingersoll-Ross  continuous-time models
have also been proposed for interest rates~\parencite{interestrate,pai}. 
A quasi-maximum likelihood estimator for a threshold diffusion with applications to interest-rate modelling is studied in \citet{su2015,su2017}.
 In~\citet{brockwell97a}, a continuous-time equivalent of an
integrated SETAR model is constructed and applied to financial data.
\citet{mota14a,esquivel14a} propose two
continuous-time models which mimic SETAR time series.  
In \citet{mota14a}, one of these models, referred to as the \emph{delay threshold regime
switching model} (DTRS), is tested on the daily prices of twenty-one companies over
almost five years. For almost all the stocks, they find a regime-change for the
volatility. 

\paragraph{Contribution of the paper. } We present the \emph{geometric
oscillating Brownian motion} (GOBM), a threshold local volatility model with
piecewise constant volatility and drift, as in \citet{Gairat:2014th}. This
model is an instance of the \emph{tiled volatility model}
of~\citet{Lipton:2011um}.  We stress that the GOBM is the solution of a
one-dimensional Stochastic Differential Equation (SDE). Therefore, it is
simpler to manipulate than the DTRS of \citet{mota14a}, although having similar
features.  For the same reason, the market is complete under the GOBM. The GOBM
can also be simulated by a standard Euler scheme \parencite{yan2002,chan1998}.
Option valuation can be performed as well using semi-analytic approaches
\parencite{Lipton:2011um,lipton:2018,Gairat:2014th,pigato,decamps:04a}, and the related
problem of estimating ex-ante volatilities from the call prices can be solved
using Sturm-Liouville theory \parencite{Lipton:2011um}. 

In the GOBM model, a fixed threshold separates two regimes for the prices. Both
the volatility and the drift parameter can assume two possible values,
according to the position of the stock price, above or below the threshold.
Let us write $\sigma_-$ for the volatility below the threshold, $\sigma_+$ for
the volatility above the threshold, and similarly $b_-$ and $b_+$ for the
drift. Such model accounts of the leverage effect, when $\sigma_->\sigma_+$. In
this case, when prices are low, volatility increases, consistently with what is
observed on empirical financial data. As in \citet{mota14a}, the dynamics has
two regimes, one corresponding to the bull market, with prices above the
threshold and low volatility, and one corresponding to the bear market, with
prices below the threshold and high volatility. In this sense, the model
displays an ``endogenous'' regime switch.  A motivation for considering such
price dynamics coming from a different viewpoint is given in
\citet{ankirchner}, where it is shown that the GOBM describes the price dynamics
corresponding to the optimal strategy for a manager who can control, in a
stylized setting, the volatility of the value of a firm, getting bonus payments
when the value process performs better than a reference index.

After describing and motivating the model, we consider the
estimation of volatilities, drifts and thresholds from discrete observations of
historical stock prices.  The estimation procedures used in
\citet{mota14a,esquivel14a,brockwell97a} are all derived from the ones designed
for SETAR time series.  Here, we approach the problem directly, proposing an
estimation procedure based on stochastic calculus.  The estimator of the
volatility coefficients is inspired by the integrated volatility/realized
variance estimator; for its theoretical analysis we refer the reader to
\citet{lejay_pigato}.  Our estimator can be implemented
straightforwardly, differently from the MLE, which is very hard to implement as
there is no simple closed form for the density of the GOBM. On the other hand,
the estimator of the drift coefficient is the maximum likelihood (MLE) one. 
Its implementation is also straightforward.
Its asymptotic behavior is studied in \citet{lejay_pigato2}.

In the present paper, we discuss several issues regarding the quality of the
estimation and propose a method for estimating the threshold, based on the
Akaike information principle. In
addition, we provide a hypothesis test to decide whether or not the volatility
is constant. We test the performance of such methods via numerical experiments
on simulated data. These tests are conclusive.

Finally, we look at empirical financial data. We first benchmark our model
against the same dataset as~\citet{mota14a}: twenty-one stock prices from the NYSE, on
the time window 2005-2009. We find similar results: in particular, we
consistently find leverage effects ($\sigma_-> \sigma_+$) and mean-reverting
behavior ($b_->0$ and $b_+<0$). Then, we apply our estimators to the empirical time
series of the S\&P 500, on the three separate five years windows $2003-2007$,
$2008-2012$ and $2013-2017$, finding again consistent evidence of leverage
effects. More specifically, we may say, based on the hypothesis test mentioned
above and on the estimated ratios $\frac{\sigma_-}{\sigma_+}$, that the
leverage effect is particularly marked in the period $2008-2012$, most likely
because it contains the $2008$ financial crisis. The mean-reverting behavior is
also quite clearly detectable in the  $2008-2012$ period, less so in the
periods $2003 - 2007$ and $2013-2017$, on which $b_-$ is always positive but $b_+$
does not display a predominant sign. This seems to be in agreement with the
finding that ``\dots the speed at which stocks revert to their fundamental
value is higher in periods of high economic uncertainty, caused by major
economic and political events'', as was shown in \citet{SPIERDIJK2012228}. We
refer to the same paper and to Section~\ref{sec:SP500} for the economic
interpretation of this finding.

As final consideration, we remark that the GOBM, despite its extreme simplicity
and limited number of parameters, reproduces notable stylized facts of
financial markets such as leverage effects and mean-reverting properties.
Moreover, the application of the estimators described above to empirical data
confirm the presence of such features in the dynamics of financial indices.

\bigskip

\noindent\textbf{Outline.} The GOBM is presented in Section~\ref{sec:gobm}.
In Section \ref{sec:estimate} we consider the estimation procedures for the volatility (Section~\ref{section:vol}), the drift (Section~\ref{section:drift}) 
and the threshold (Section~\ref{section:threshold}).
In Section~\ref{sec:empirical}, we benchmark the GOBM model against 
the DTRS model (in Section~\ref{sec:dtrs}) introduced by \citet{mota14a}
by comparing the estimators on the same data sets
(in Section~\ref{sec:estimation_gobm}).
In Section~\ref{sec:leverage}, we present a hypothesis test to 
decide whether a leverage effect is present or not.
Finally, in Section \ref{sec:SP500} we apply our estimators to the stock prices of the S\&P 500, 
in three consecutive periods of five years, the second one containing
the 2008 financial crisis. 
The article ends with a global conclusion in Section~\ref{sec:conclusion}. 


\section{A threshold model for local volatility}

\label{sec:gobm}

\paragraph*{The model.}
The geometric oscillating Brownian motion (GOBM) is the solution to the local volatility model 
\begin{equation}
    \label{eqnprice}
S_t=x+\int_0^t \sigma(S_s)S_s\vd B_s+\int_0^t \mu(S_s)S_s\vd s,
\end{equation}
where $B$ is a Brownian motion and for a threshold $m\in\RR$, 
\begin{equation}
    \label{parametersGOBM}
    \sigma(x)=\begin{cases}
	\sigma_+&\text{ if }x\geq m,\\
	\sigma_-&\text{ if }x<m
    \end{cases}
    \text{ and }
    \mu(x)=\begin{cases}
	\mu_+&\text{ if }x\geq m,\\
	\mu_-&\text{ if }x<m.
    \end{cases}
\end{equation}

We use a solution $S$ to \eqref{eqnprice} as a model 
for the price of an asset. The \emph{log-price} $X=\log(S)$ satisfies the SDE
\begin{equation}
    \label{eq:logprice}
    X_t=x+\int_0^t \sigma(X_s)\vd B_s+\int_0^t b(X_s)\vd s,
\end{equation}
with
\begin{equation}
    \label{parametersOBM}
    \sigma(x)=\begin{cases}
	\sigma_+&\text{ if }x\geq r,\\
	\sigma_-&\text{ if }x<r
    \end{cases}
    \text{ and }
    b(x)=\begin{cases}
	b_+=\mu_+-\sigma_+^2/2&\text{ if }x\geq r,\\
	b_-=\mu_- -\sigma_-^2/2&\text{ if }x<r
    \end{cases}
\end{equation}
for a threshold $r=\log(m)$. 
Notice the slight abuse of notation in \eqref{parametersGOBM} and \eqref{parametersOBM}, due to the change of the value for the threshold when taking the logarithm. 

When the drift $b=0$ and $r=0$, $X$ is called an \emph{oscillating Brownian motion} (OBM, \citet{keilson}), a name we keep even in presence of a two-valued drift and a threshold $r \neq 0$. 
When $\sigma_+=\sigma_-$ and $b_+=b_-$, the price follows the Black-Scholes model. 
By extension, we still call the solution to \eqref{eq:logprice} a GOBM.

The effect of the drift is discussed in Section~\ref{section:drift}. 
When $b_+<0$ and $b_->0$, the process is ergodic and mean-reverting. The convergence towards equilibrium 
differs from the ones in the Va\v si\v ceck and Heston models in which the drift is linear.

\paragraph*{Existence and uniqueness.} 
The solution to \eqref{eq:logprice} is an instance of a more general 
class of processes with discontinuous coefficients which was studied
in~\citet{legall}. In particular,
there exists a unique strong solution to \eqref{eq:logprice}, 
hence to \eqref{eqnprice}.

The (geometric)-OBM can be easily
manipulated with the standard tool of stochastic analysis, sometimes relying
on the Itô-Tanaka formula instead of the sole Itô formula \parencite[see \textit{e.g.,}][]{etore}.

\paragraph*{Properties of the market.} 
Unlike in some regime switching models, there is no hidden randomness leading to incomplete markets,
while offering some regime change properties. 

\begin{proposition} Assuming the GOBM model for the returns process with a constant risk-free rate, 
    the market is viable and complete.
\end{proposition}
\begin{proof} Using the results of \citet{legall}, the Girsanov theorem can be applied to the equation for
    the log-price. Hence, as for the Black-Scholes model, it is possible
    to reduce the discounted log-price to a martingale by removing the drift. Hence, there 
    exists an equivalent martingale measure, meaning that the market is viable
    \parencite[Theorem~2.1.5.4, p.~89]{jeanblanc}.

    As any absolutely continuous measure could only be reached through 
    a Girsanov transform \parencite{legall}, the risk neutral measure is unique, 
    meaning that the market is complete.
\end{proof}
\begin{remark} 
    \label{rem:sbm}
    The affine transform $\Phi(x)=x/\sqrt{\sigma_-}\ind{x<m}+x/\sqrt{\sigma_+}\ind{x>m}$
    transforms the log-price $X_t$ into the solution $Y$ to the SDE with local 
    time 
    \begin{equation}
	\label{eq:sbm}
	Y_t=\Phi(X_0)+B_t+\int_0^t \frac{b_+}{\sigma_+}\ind{x>0}(Y_s)\vd s
	+\int_0^t \frac{b_-}{\sigma_-}\ind{x<0}(Y_s)\vd s
	+\kappa L_t^0(Y),
    \end{equation}
    where $(L_t^0(Y))_{t\geq 0}$ is the local time of $Y$ at position $0$
    and $\kappa=(\sqrt{\sigma_-}-\sqrt{\sigma_+})/(\sqrt{\sigma_-}+\sqrt{\sigma_+})$.
    Eq.~\eqref{eq:sbm} is a drifted skew Brownian motion (SBM). 
    The local time part cannot be removed by a Girsanov transform 
    \parencite{legall}.  For the SBM, arbitrage may exist as shown in \citet{rossello}.
    The GOBM may be generalized by considering
    a log-price solution to $\vd X_t=x+\sigma(X_t)\vd B_t
    +b(X_t)\vd t+\eta \vd L_t^r(X)$, where $L^r(X)$ is the local time of $X$
    at the threshold $r$. The effect of the coefficient $\eta\in(-1,1)$
    would be to ``push upward'' (if $\eta>0$) or downward (if $\eta<1$)
    the price, which corresponds to some directional predictability effect \parencite{alvarez}. 
    However, considering $\eta\not=0$ radically changes the structure of the market
    with respect to classical SDEs.
\end{remark} 

\paragraph*{Monte Carlo simulation.}
The GOBM at times $kT/n$, $k=0,1,2,\dotsc,n$ is easily simulated
by $\overline{S}_{kT/n}=\exp(\overline{X}_{kT/n})$ through 
the recursive equation \parencite{chan1998,yan2002}
\begin{equation}
    \overline{X}_{\frac{(k+1)T}{n}}
    =\overline{X}_{\frac{kT}{n}}+\sqrt{\frac{T}{n}}\sigma\left(\overline{X}_{\frac{kT}{n}}\right)\eta_k
    +b\left(\overline{X}_{\frac{kT}{n}}\right)\frac{T}{n},\ 
    \eta_k\sim\cN(0,1)\text{ independent}. 
\end{equation}

\section{Estimation of the parameters from the observations of the stock prices}

\label{sec:estimate}
The GOBM $X$ is defined by five parameters (volatility, drift and threshold, 
see Table~\ref{table:notation})
which we are willing to estimate. 
In Sections~\ref{section:vol} and
\ref{section:drift} we consider the estimation of $(\sigma_\pm,b_\pm)$ for
fixed threshold $r$, by considering the estimation of the ex-post volatility
and of the drift. Afterwards, in Section~\ref{section:threshold}, 
the threshold is chosen through a model selection principle.

The procedure presented here is simple to implement and provides good results
in practice. We stress that the estimators of $\sigma_\pm$ can be implemented
with no previous knowledge of the drift, and viceversa the estimators of
$b_\pm$ do not need the knowledge of $\sigma_\pm$ to be implemented.  The
estimators of $\sigma_\pm$ are integrated volatility type estimators, simple to
implement and widely studied in the framework of SDEs. One could also use a MLE
for the volatility, based on discrete observations, but this would require
explicit expressions for the transition density, which are very involved for
the GOBM. Concerning the estimator for the drift, our MLE (cf.~\eqref{estb})
can be implemented  using estimators for local and occupation times proposed in
\citet{lejay_pigato2}, which do not involve the volatility parameter.

\begin{table}[ht]
    \begin{center}
	\begin{tabular}{l l }
	    \toprule
	    $S$ & price of the stock  \\
	    $X=\log(S)$ & log-price\\
	    $\xi=X-r$ & shifted log-price for a threshold $r$ \\
	    \midrule
	    $r$ & threshold of $X$, the log-price \\
	    $m=\exp(r)$ & threshold of $S$, the price \\
	    $\sigma_-$ & volatility of $X$ below $r$  \\
	    $\sigma_+$ & volatility of $X$ above $r$  \\
	    $b_-$ & drift of $X$ below $r$  \\
	    $b_+$ & drift of $X$ above $r$  \\
	    $\mu_-=b_-+\frac{\sigma_-^2}{2}$ & appreciation rate of $S$ below $m$  \\
	    $\mu_+=b_++\frac{\sigma_+^2}{2}$ & appreciation rate of $S$ above $m$  \\
	    \midrule
	    $d$ & delay (DTRS only)\\
	    \bottomrule
	\end{tabular}
	\caption{\label{table:notation} Notations for the GOBM and DTRS models.}
    \end{center}
\end{table}

\subsection{Estimation of the ex-post volatility}
\label{section:vol}

In this section we consider the estimation of the ex-post volatility for prices given
by the model in \eqref{eqnprice}, when the threshold $r=\log(m)$ is known. We recall the
estimators and the theoretical convergence results presented in
\citet{lejay_pigato}, and discuss their application in the framework of
volatility modeling. 
We define $\xi\eqdef X-r=\log(S)-r$ which is a OBM with a threshold at level $0$.

\paragraph*{The data.}
Our observations are $n+1$ daily data $\Set{\xi_{k}}_{k=0,\dotsc,n}$
with $\xi_k=\log(S_k)-r$ for an \textit{a priori} known threshold $r$.

Our aim is to estimate $(\sigma_+,\sigma_-)$ from such observations.

\paragraph*{Occupation times.}
The occupation times below and above the threshold play a central role in our study. 
The positive and negative
\emph{occupation times} $Q^\pm_T$ up to time $T$ of $\xi$ 
as well as their approximations $\Obs{Q}_\pm(n)$ are
\begin{equation}
    \label{eq:drift:1}
    Q^\pm_T:=\int_0^T \mathbf{1}_{\pm\xi_s\geq 0}\vd s
    \text{ and }
    \Obs{Q}_\pm(n)\eqdef\frac{T}{n}\sum_{i=1}^{n} \mathbf{1}_{\pm\xi_{i}\geq 0}.
\end{equation}

\paragraph*{The estimators of the volatilities.}
We write $\xi^+:=\max\Set{\xi,0}$ and $\xi^-:=-\min\Set{\xi,0}$, 
the positive and negative parts of $\xi$.
Our estimators $\upsigma_\pm(n)^2$ for~$\sigma_\pm^2$ are 
\begin{equation}
    \upsigma_\pm(n)^2:=
    \frac{\displaystyle
	[\xi^\pm,\xi]_n
    }
    {\displaystyle
	\Obs{Q}_\pm(n)
    },
\end{equation}
where 
\begin{equation*}
    [\xi^\pm,\xi]_n\eqdef \sum_{i=1}^n (\xi^\pm_{i}-\xi^\pm_{i-1})(\xi_{i}-\xi_{i-1}).
\end{equation*}
These estimators are natural generalizations of the 
\emph{realized volatility estimators} \parencite{barndorff2002}.

\begin{proposition}[\cite{lejay_pigato}] 
    \label{prop:estim}
    When $b_+=b_-=0$, 
    the couple $(\upsigma_-(n)^2,\upsigma_+(n)^2)$ is a consistent estimator
    of $(\sigma_-^2,\sigma_+^2)$. Besides, there exists a pair
    of unit Gaussian random variables $(G_-,G_+)$ independent
    of the underlying Brownian motion $B$ (hence of $\xi$) such that 
\begin{equation}
    \label{conv:sigma}
    \begin{bmatrix}
	\sqrt{n}\sqrt{\dfrac{\Obs{Q}_-(n)}{n}}\left(\upsigma_-(n)^2-\sigma_-^2\right)\\
	\sqrt{n}\sqrt{\dfrac{\Obs{Q}_+(n)}{n}}\left(\upsigma_+(n)^2-\sigma_+^2\right)
\end{bmatrix}
    =
    \begin{bmatrix}
    \dfrac{[\xi^-,\xi]_n- \Obs{Q}_-(n) \sigma_-^2}{\sqrt{\Obs{Q}_-(n)}}
    \\
\dfrac{[\xi^+,\xi]_n- \Obs{Q}_+(n) \sigma_+^2}{\sqrt{\Obs{Q}_+(n)}}
\end{bmatrix}
\xrightarrow[n\to\infty]{\text{law}}
    \begin{bmatrix}
	\sqrt{2}\sigma_-^2 G_-\\
    \sqrt{2}\sigma_+^2 G_+
\end{bmatrix}.
\end{equation}
\end{proposition}


\paragraph*{Dealing with a drift.} Proposition~\ref{prop:estim} is actually 
proved on high-frequency data $\xi_{k,n}:=\xi_{k/n}$, $k=0,\dotsc,n$ on the time interval $[0,1]$. 

Using a scaling argument, for any constant~$c>0$, 
$\Set{c^{-1/2}\xi_{ct}}_{t\geq 0}$ is equal in distribution
to $\zeta^{(c)}$ solution to the SDE
\begin{equation}
    \vd \zeta^{(c)}_t=\sigma(\zeta^{(c)}_t)\vd W_t+\sqrt{c} b(\zeta^{(c)}_t)\vd t,\ \zeta^{(c)}_0=c^{-1/2}\xi_0
\end{equation}
for a Brownian motion $W$.  With $c=n$, the problem of estimating 
the coefficients of 
$\Set{\xi_{k}}_{k=0,\dotsc,n}$ is the same as the \emph{high frequency} estimation 
of the coefficients of $\Set{\sqrt{n}\zeta^{(n)}_{k,n}}_{k=0,\dotsc,n}$ on
the time range $[0,1]$. 

Without drift, observing $\Set{\xi_{k,n}}_{k=0,\dotsc,n}$
or $\Set{\xi_{k}}_{k=0,\dotsc,n}$ leads to the same estimation.
Using the Girsanov theorem, Proposition~\ref{prop:estim} stated for the high-frequency regime, 
that is on the observations $\Set{\zeta^{(1)}_{k,n}}_{k=0,\dotsc,n}$ 
(since all the $\zeta^{(n)}$ are equal in distribution), is also 
valid in presence of a bounded drift.

With our data, the drift is very small compared to the ex-post volatility and the number~$n$ 
of observations is finite, so that we still apply Proposition~\ref{prop:estim}.


\subsection{Estimation of the drift coefficients}
\label{section:drift}

To estimate the values $b_\pm$ of the drift,
we consider that the threshold $r=\log m$ is known
(this issue is treated in Sect.~\ref{section:threshold}).

\paragraph*{Maximum likelihood estimation of the drift.}
Following \citet{lejay_pigato2}, we introduce as estimators of $b_\pm$ the quantities
\begin{equation}
    \label{estb}
    \beta_\pm(T)\eqdef \pm
 \frac{\xi_T^\pm-\xi_0^\pm -L_T/2 }{Q^\pm_T(\xi)}
 \text{ and }
    \upbeta_\pm(n)\eqdef \pm \frac{\xi_T^\pm-\xi_0^\pm -\Obs{L}(n)/2 }{\Obs{Q}_\pm(n)},
\end{equation}
where $L_T$ is the symmetric local time of $\xi$ at $0$ while 
$\Obs{L}(n)$ is the discrete local time defined by 
\begin{equation*}
    \Obs{L}(n)\eqdef \sum_{i=0}^{n-1} \ind{\xi_i\xi_{i+1}<0}\abs{\xi_{i+1}}. 
\end{equation*}
The expressions of $\Obs{L}(n)$ and $\upbeta_\pm(n)$ do not involve $\sigma_\pm$.
Besides, $\upbeta_\pm(n)$ converges almost surely to $\beta_\pm(T)$ as $n$ goes to infinity.

When the coefficients are constant, as for the log-price in the Black-Scholes
model, where $\vd \zeta_t=\sigma \vd B_t+ b\vd t$, a consistent estimator of the drift $b$
is $b(T)=(\zeta_T-\zeta_0)/T$.  
Our estimator~\eqref{estb} generalizes this formula; 
the local time term appears because of the discontinuity in the coefficients.

\paragraph{Asymptotic properties.}

The drift estimator shall be studied for a long time horizon.
The asymptotic properties of $\beta_\pm(T)$ as $T\to\infty$, hence of $\upbeta_\pm(n)$,
depend on the asymptotic behaviors of $Q^\pm_T$ in \eqref{eq:drift:1}. 
We summarize in Table~\ref{table:drift} the different cases that
depend solely on the respective signs of $b_+$ and $b_-$.

\begin{table}[!ht]
    \begin{center}
	\begin{tabular}{c | c c  c}
\multicolumn{1}{c}{} & $b_+<0$ & $b_+=0$ &  $b_+>0$ \\
	\toprule
	$b_->0$ & {ergodic (\textbf{E})} & null recurrent (\textbf{N1}) & transient (\textbf{T0})  \\
	$b_-=0$ & null recurrent (\textbf{N1}) & null recurrent (\textbf{N0}) & transient (\textbf{T0})\\
	$b_-<0$ & transient (\textbf{T0}) & transient (\textbf{T0}) & transient (\textbf{T1}) \\
	\bottomrule
    \end{tabular}
    \caption{\label{table:drift} Regime of $\xi$ according to the respective signs
    of $(b_-,b_+)$.}
\end{center}
\end{table}

The ergodic case, which corresponds to a mean-reverting process, is of course
the most favorable one. In the transient case, the estimators may not converge. 
We briefly recall some of the results in \citet{lejay_pigato2} (see the original paper for detailed statements).

\begin{itemize}[leftmargin=1cm]
    \item[\textbf{E}] The ergodic case is the mean-reverting one. 
	In this case ($\beta_-(T),\beta_+(T))$ converges almost surely to $(b_-,b_+)$, with speed $\sqrt{T}$.

    \item[\textbf{T0}] If $b_+>0$, $b_-\geq 0$, then $\beta_+(T)$ converges
	to~$b_+$ with speed $\sqrt{T}$.
	The estimator $\beta_-(T)$~of~$b_-$ does not converge to $b_-$
	and is then meaningless.
	The case $b_-<0$, $b_+\leq 0$ is treated by symmetry.

    \item[\textbf{T1}] If $b_+>0$ and $b_-<0$, then 
	with probability $p\eqdef\sigma_-b_+/(\sigma_+b_-+\sigma_-b_+)$,
	$\beta_+(T)$ converges to $b_+$ with speed $\sqrt{T}$,  
	while with probability $1-p$, 
	$\beta_-(T)$ converges to $b_-$ with speed $\sqrt{T}$.
	This asymptotic behavior is due to the fact that after a given 
	random time, the process does not cross the threshold anymore. 

    \item[\textbf{N0}] Whatever $T>0$ the distribution of $\sqrt{T}(\beta_-(T),\beta_+(T))$ does not depend 
	on~$T$.  Then~$\beta_\pm(T)$ are consistent estimators of $b_\pm=0$.

    \item[\textbf{N1}] If $b_+=0$, $b_->0$, then $(\beta_-(T),\beta_+(T))$ converges
	almost surely to $(b_-,b_+)$; $\beta_-(T)$ converges to $b_-$ with speed $T^{1/4}$ and
		$\beta_+(T)$ converges to $b_+$ with speed~$\sqrt{T}$.
\end{itemize}

\subsection{Estimation of the threshold}
\label{section:threshold}

The above estimators for $\sigma$ and $b$ assume that the value $m$ of the threshold is known. 
Following \citet{tong1983} \parencite[see also][p.~79]{priestley}, we estimate $m$ using a principle
of \emph{model selection} relying on the ideas of the Akaike information principle (AIC) \parencite{akaike}. 
Since the AIC involves the likelihood function, for which we do not necessarily have closed form expressions, we will need to work with approximations.

\paragraph{Approximation of the density.}

Given a threshold $r$ as well as volatility and drift functions $x\mapsto \sigma(x)$ and $x\mapsto b(x)$, 
we first consider the density $y\mapsto p(\Delta t,x,y;r,\sigma,b)$ of $X_{t+\Delta t}$ given $X_t=x$
(the process is time-homogeneous so that $p$ only depends on $\Delta t$, not on $t$).
For a vanishing drift, a closed form expression for $p$ is known \parencite{keilson}. 
In presence of a drift, the expression may become cumbersome if not intractable \parencite{lejay2015}.
However, $p$ can be approximated, in short time, via the related Green function, easier to compute \parencite[see][Chapter~2]{lenotre}.
Alternatively, we assume that the drift is constant over the time interval $[t,t+\Delta t]$
and replace $p$ by the density of $Y_{t+\Delta t}+b(x)\Delta t$ given $Y_t=x$, where $Y$
has the same volatility of $X$ yet with a vanishing drift. 
In the implementation, we use the latter approximation of $p$
which we denote by $\widetilde{p}(t,x,\cdot;r,\sigma,b)$. 

\paragraph{Selection of the threshold.}
The procedure to select the ``best'' threshold is then 
\begin{enumerate}[nosep,label=\arabic*/]
    \item We fix $r^{(1)},\dotsc,r^{(k)}$ possible thresholds in the range of the observed values $\Set{X_{t_i}}_{i=0,\dotsc,T}$ of the log-price $X$.
    \item For each threshold $r^{(j)}$, we estimate the drift and volatilities $\widehat{\sigma}^{(j)}$ and $\widehat{b}^{(j)}$.
    \item We compute the approximate log-likelihood 
	\begin{equation}
	    \label{eq:loglik}
	    \LogLik(j)=\sum_{i=0}^{T-1}\log \widetilde{p}(t,X_{t_i},X_{t_{i+1}};r^{(j)},\widehat{\sigma}^{(j)},\widehat{b}^{(j)}).
	\end{equation}
    \item We select as threshold $\widehat{r}$ the value $r^{(\tilde{\jmath})}$ where 
	$\tilde{\jmath}\eqdef\argmin_{j=1,\dotsc,k} \LogLik(j)$.
\end{enumerate}

\paragraph{Comparison with other models.}
In the model selection based on the AIC, the best model is the one for which
the log-likelihood corrected by a value depending on the number of parameters
is minimized. Here,  the number of parameters is fixed to $4$ so that it is sufficient to
use only approximations of the log-likelihoods.  A similar procedure is used in
\citet{meng13a}, yet with a density estimated through Monte Carlo, which is
time-consuming. On the contrary, our procedure avoids any simulation step.
With respect to the estimation for the SETAR model
\parencite{tong1983,priestley}, as well as the one of the DTRS model presented below,
based on least squares \parencite{mota14a}, there is no delay so that the dimension
of the model is reduced by $1$.


\section{Benchmarking the model}

\label{sec:empirical}

We apply now our estimators to empirical financial data. We benchmark our model
against the \emph{delay and threshold regime switching model} (DTRS)
of \citet{mota14a} by using the same data. We start this section shortly
presenting the DTRS model.

\subsection{The delay and threshold regime switching model}

\label{sec:dtrs}

\citet{mota14a} introduce the DTRS. First, they consider two 
sets of (functional) parameters $(\sigma_1,\mu_1)$ and $(\sigma_2,\mu_2)$,
as well as a diffusion solution to the stochastic differential equation 
\begin{equation}
    \vd S_t=\mu_{J_t}(t,S_t)\vd t+\sigma_{J_t}(t,S_t)\vd B_t
\end{equation}
for a Brownian motion $B$, where $J$ is a non-anticipative process
with values in the set of indices $\Set{1,2}$. 

The rule for $J$ to switch is based on a threshold $m$, a delay $d$ as well as
a small parameter $\epsilon>0$.
Assume $S_0\leq m$ and $J_0=1$. The process evolves according
to the parameters $(\sigma_1,\mu_1)$ until it reaches the level $m+\epsilon$
at a (random) time $\tau_1$. Then it evolves according 
to the parameters $(\sigma_1,\mu_1)$ up to time $\tau_1+d$ where it
switches to parameters $(\sigma_2,\mu_2)$ (that is $J_{\tau_1+d}=2$) 
until it reaches the level $m$ at time $\tau_2$. Then it switches again to the state $1$ 
after a delay $d$ ($J_{\tau_2+d}=1$)
and so on. 

The parameter $\epsilon$ prevents an accumulation of \textquote{immediate} switches
so that~$S$ can be constructed on a rigorous basis \parencite{esquivel14a}.
With respect to simulation or estimation, $\epsilon$ is of no practical importance 
as~$S$ is only observed or simulated at discrete times. 

More specifically, the DTRS model considered in \citet{mota14a} assumes that the~$\mu_i$ and~$\sigma_i$ ($i=1,2$) are
\begin{equation}
    \begin{cases}
	\sigma_1(t,x)=\sigma_-\cdot x & \text{ if }x<m,\\
	\sigma_2(t,x)=\sigma_+\cdot x & \text{ if }x\geq m
    \end{cases}
    \text{ and }
    \begin{cases}
	\mu_1(t,x)=\mu_-\cdot x & \text{ if }x<m,\\
	\mu_2(t,x)=\mu_+\cdot x & \text{ if }x\geq m
    \end{cases}
\end{equation}
for some constants $\sigma_\pm>0$ and $\mu_\pm$. Hence, on each
regime, the price $S$ follows a dynamic of Black-Scholes type. 
We also define
$b_\pm=\mu_\pm-\sigma_\pm^2/2$ so that $b_\pm$ are the possible
values of the drift for the log-price.

Adapting the estimation approach for the SETAR \parencite{tong1983},
\citet{mota14a} propose a consistent estimation procedure of the parameters, based on least squares.

\paragraph{Results for the DTRS.}
This estimator is applied to the daily log-prices of twenty-one~stock
prices of the NYSE, from January 2005 to November 2009 (presented in Table~\ref{table:stock}). In Table~\ref{table:dtrs}, we report
the estimated values of $\sigma_\pm$, $m$, $\mu_\pm$ (or $b_\pm$) and $d$
found in~\citet{mota14a}. These values have to be compared
with the ones in Table~\ref{table:gobm}.

\begin{table}
    \begin{center}
\begin{tabular}{rl rl rl}
\toprule
AAPL &  Apple &
ADBE &  Adobe \\  
AMZN &  Amazon &
C &  CitiGroup \\
CA &  CA & 
CSCO &  Cisco \\
GOOG &  Google & 
HP &  Hewlett-Packard \\  
IBM &   IBM &
JPM &  JP Morgan \\
KO &  Coca-cola & 
MCD &  McDonalds \\ 
MON &  Monsanto &
MSFT &  Microsoft\\
MSI &  Motorola & 
NYT &   New-York Times  \\
PCG &  PG\&E & 
PFE &  Pfizer \\ 
PG &  P \& G &
PM &  Philip Morris  \\ 
SBUX &  Starbucks \\
\bottomrule
\end{tabular}
    
    \caption{\label{table:stock} Abbreviations of the names of the stocks
    (in Yahoo Finance).}
\end{center}
\end{table}

For most of the data, a leverage effect is observed: 
the ex-post volatility below the threshold is higher than above it. 
In \citet{mota14a}, option prices of European calls are also computed using a Monte Carlo procedure.  The resulting prices are in good agreement with the ones of the market.

\begin{table}[t]
    \centering\small
    \begin{tabular}{rrS[table-format=3.1]S[table-format=3.1]@{\,}S[table-format=3.1]S[table-format=+3.1]S[table-format=+3.1]S[table-format=+3.1]S[table-format=+3.1]c}
\toprule
\multicolumn{9}{c}{Delay threshold regime switching (DTRS) \parencite{mota14a}}\\
\toprule
\multicolumn{1}{c}{Index} & {$d$} & {$m$ [\$]} & {$\sigma_-$ [\%]} & {$\sigma_+$ [\%]} & {$\mu_-$ [\%]} & {$\mu_+$ [\%]} &  {$b_-$ [\%]} & {$b_+$ [\%]} & {signs} \\
\cmidrule{2-10}
AAPL & 8 & 173.5 & 54.1 & 45.6 & 57.5 & -112.4 & 42.8 & -122.7 & $+-$ \\
ADBE & 1 & 41.5 & 44.0 & 25.1 & 31.8 & -74.3 & 21.9 & -77.4 & $+-$ \\
AMZN & 1 & 77.7 & 51.4 & 45.4 & 54.2 & -462.2 & 41.1 & -472.5 & $+-$ \\
C & 2 & 43.1 & 120.6 & 16.8 & 39.8 & -3.5 & -32.8 & -5.0 & $--$ \\
CA & 1 & 22.1 & 53.0 & 24.6 & 44.4 & -23.2 & 30.2 & -26.2 & $+-$ \\
CSCO & 1 & 16.3 & 56.0 & 31.3 & 313.5 & -0.0 & 297.6 & -5.0 & $+-$ \\
GOOG & 13 & 642.0 & 37.0 & 40.2 & 40.6 & -148.7 & 33.8 & -156.7 & $+-$ \\
HP & 1 & 46.9 & 39.2 & 27.1 & 42.3 & -78.9 & 34.8 & -82.4 & $+-$ \\
IBM & 1 & 124.3 & 25.1 & 20.3 & 12.9 & -93.5 & 9.6 & -95.5 & $+-$ \\
JPM & 2 & 25.0 & 131.3 & 47.5 & 715.4 & -0.3 & 629.2 & -11.6 & $+-$ \\
KO & 1 & 10.0 & 78.6 & 29.1 & 398.2 & -5.3 & 367.4 & -9.6 & $+-$ \\
MCD & 1 & 54.6 & 22.4 & 25.7 & 38.8 & -29.5 & 36.3 & -32.8 & $+-$ \\
MON & 1 & 112.0 & 47.0 & 49.4 & 57.5 & -145.4 & 46.6 & -157.5 & $+-$ \\
MSFT & 1 & 22.9 & 54.0 & 25.1 & 81.1 & 6.0 & 66.5 & 3.0 & $++$ \\
MSI & 14 & 21.9 & 49.5 & 28.4 & 7.8 & -47.9 & -4.3 & -51.9 & $--$ \\
NYT & 4 & 32.5 & 49.5 & 17.3 & -9.8 & -89.5 & -22.2 & -91.0 & $--$ \\
PCG & 6 & 35.3 & 49.4 & 21.6 & 170.1 & -4.0 & 158.0 & -6.6 & $+-$ \\
PFE & 2 & 16.7 & 40.2 & 24.0 & 67.0 & -14.6 & 59.0 & -17.4 & $+-$ \\
PG & 1 & 61.9 & 20.3 & 20.2 & 19.2 & -28.0 & 17.1 & -30.0 & $+-$ \\
PM & 1 & 42.0 & 44.4 & 31.9 & 121.2 & -40.3 & 111.4 & -45.4 & $+-$ \\
SBUX & 15 & 33.6 & 45.4 & 26.0 & 11.6 & -39.8 & 1.3 & -43.1 & $+-$ \\
\bottomrule
\end{tabular}
\caption{\label{table:dtrs}
Estimated daily parameters in \% per year found by \citet{mota14a}
for the DTRS model on daily data from January 2005 to November 2009
with the notations given in Table~\ref{table:notation}
(In the original table, volatilities and drift are expressed in \% per day).
The last column \emph{signs} contains the respective signs of $b_-,b_+$
(a $+-$ indicates a mean-reversion effect). }
\end{table}

\paragraph*{Comparison between the DTRS model and the GOBM.}
In spirit, the GOBM is similar to DTRS of \citet{mota14a} or to the model in \citet{esquivel14a}. 
Yet, it avoids all the difficulties related
to the ``gluing'' and regime change that involves a very thin layer
which serves to avoid infinitely many immediate switches. 
\citet{Hottovy} discuss the asymptotic behavior of the process
as the width of the layer decreases to $0$.

The GOBM has five parameters while the DTRS has six parameters because it also involves a delay. 
For most of the data, the estimated delay in the DTRS is $d=1$, 
which means that the switching occurs without delay. Otherwise,
the delay means a slow decreasing auto-correlation, or a long memory effect. 
Yet, for long delay, how to discriminate a leverage effect 
from sudden changes due to external parameters such as crisis?
The presence of a delay increases the possibility of miss-specifications in the estimation procedure. 

\subsection{Estimation of the parameters of the GOBM}
\label{sec:estimation_gobm}


In Table~\ref{table:gobm}, we estimate the parameters for the GOBM 
on the same stocks as for the DTRS. 
Although we use the same source (Yahoo Finance) as \citet{mota14a}, 
it seems that KO is a different time series than in this article.

The volatilities $(\sigma_-,\sigma_+)$ are in good agreement for both
models. Moreover, $\sigma_- >\sigma_+$ for all the stocks, 
the only exceptions being MCD for both models and GOOG for the DTRS model.
In the latter situation, $\sigma_-$ is close to $\sigma_+$. 
The respective signs of~$b_-$ and~$b_+$ are consistent with the ones of \citet{mota14a}
and suggest a mean-reversion effect ($b_->0$, $b_+<0$) for most of the stock prices.
The magnitudes of $b_-$ and $b_+$ are also consistent. As the number of data is rather small 
($n=\num{1217}$) and the considered period is only five years, 
it is not reasonable to aim for a more accurate description of the drift.
Anyway, this indicates that below the threshold the ex-post volatility is higher and the drift
is upward oriented. 
The threshold estimations are in good agreement for twelve stocks out of twenty-one.

\begin{table}[t]
\centering
\begin{tabular}{cS[table-format=3.1]S[table-format=1.2]S[table-format=1.2]S[table-format=+2.2]S[table-format=+2.2]S[table-format=+2.2]S[table-format=+2.2]c}
  \toprule
Index & {$m$ [\$]} & {$\sigma_-$ [\%]} & {$\sigma_+$ [\%]} & {$\mu_-$ [\%]} & {$\mu_+$ [\%]} & {$b_-$ [\%]} & {$b_+$ [\%]} & {signs} \\ 
  \midrule
AAPL & 119.4 & 59.86 & 40.48 & 43.06 & 19.66 & 25.14 & 11.46 & $++$ \\ 
  ADBE & 46.1 & 46.79 & 81.37 & 13.93 & -136.87 & 2.98 & -169.97 & $+-$ \\ 
  AMZN & 39.7 & 38.71 & 54.57 & 36.72 & 36.59 & 29.23 & 21.70 & $++$ \\ 
  C & 40.1 & 118.58 & 17.35 & -37.30 & -7.61 & -107.60 & -9.11 & $--$ \\ 
  CA & 21.5 & 51.27 & 25.54 & 40.35 & -10.56 & 27.21 & -13.82 & $+-$ \\ 
  CSCO & 16.9 & 60.48 & 30.70 & 156.08 & 0.05 & 137.79 & -4.67 & $+-$ \\ 
  GOOG & 373.8 & 44.69 & 33.03 & 86.76 & 6.06 & 76.77 & 0.60 & $++$ \\ 
  HP & 57.6 & 66.18 & 41.09 & 53.51 & -125.14 & 31.61 & -133.58 & $+-$ \\ 
  IBM & 115.2 & 26.04 & 20.21 & 17.86 & -20.97 & 14.47 & -23.01 & $+-$ \\ 
  JPM & 32.5 & 130.32 & 41.54 & 233.41 & -2.13 & 148.48 & -10.76 & $+-$ \\ 
  KO & 47.8 & 23.70 & 17.88 & 11.75 & 5.72 & 8.94 & 4.12 & $++$ \\ 
  MCD & 51.4 & 20.23 & 28.02 & 29.42 & 3.54 & 27.37 & -0.39 & $+-$ \\ 
  MON & 85.8 & 52.71 & 55.80 & 63.67 & -99.35 & 49.78 & -114.92 & $+-$ \\ 
  MSFT & 23.2 & 51.20 & 25.90 & 74.67 & -6.49 & 61.56 & -9.84 & $+-$ \\ 
  MSI & 14.2 & 66.44 & 25.99 & -7.71 & -1.10 & -29.79 & -4.48 & $--$ \\ 
  NYT & 16.0 & 78.46 & 25.85 & -7.40 & -25.26 & -38.18 & -28.60 & $--$ \\ 
  PCG & 33.3 & 127.09 & 23.29 & 504.83 & 4.24 & 424.06 & 1.53 & $++$ \\ 
  PFE & 19.1 & 39.68 & 20.55 & 11.92 & -7.54 & 4.05 & -9.65 & $+-$ \\ 
  PG & 51.9 & 29.62 & 20.12 & 39.25 & 3.14 & 34.87 & 1.11 & $++$ \\ 
  PM & 41.9 & 45.37 & 30.81 & 141.64 & -45.60 & 131.35 & -50.35 & $+-$ \\ 
  SBUX & 13.0 & 71.88 & 46.51 & 49.37 & -13.51 & 23.54 & -24.33 & $+-$ \\ 
   \bottomrule
\end{tabular}
\caption{\label{table:gobm}%
    Estimated parameters in \% per year for the GOBM model 
on the daily data from January 2005 to November 2009
with the notations given in Table~\ref{table:notation}. 
The last column \emph{signs} contains the respective signs of $b_-,b_+$
(a $+-$ indicates a mean-reversion effect).}
\end{table}


\section{Is there some leverage effect?}
\label{sec:leverage}

Our aim is to test whether or not $\sigma_+=\sigma_-$ when $b_-=b_-=0$
(on daily data, $b_-$ and $b_+$ have small values with respect to $\sigma_-$ and $\sigma_+$).
Our Hypothesis test is then 
\begin{itemize}[nosep,partopsep=0pt,leftmargin=2cm]
    \item[$(H_0)$] (null hypothesis) $\sigma_-=\sigma_+$ ;
    \item[$(H_1)$] (alternative hypothesis) $\sigma_-\not= \sigma_+$.
\end{itemize}

\subsection{Construction of a confidence region}
The asymptotic result \eqref{conv:sigma} of Proposition~\ref{prop:estim} is rewritten as 
\begin{equation}
    \label{eq:limit}
    \begin{bmatrix}
	\upsigma_-(n)^2\\
	\upsigma_+(n)^2
    \end{bmatrix}
    \approx
    \begin{bmatrix}
	\sigma_-^2\\
	\sigma_+^2
    \end{bmatrix}
    -
    \frac{1}{\sqrt{n}}M_T\mathbf{G}
    \text{ with }M_T=\sqrt{2T}\begin{bmatrix}
    \sigma^2_-/\sqrt{Q^-_T} & 0 \\
    0 & \sigma^2_+/\sqrt{Q^+_T}
    \end{bmatrix},
\end{equation}
where $\mathbf{G}\sim\mathcal{N}(0,\mathrm{Id})$ is a Gaussian vector independent 
of the observed process~$\xi$. 
The limit term in \eqref{eq:limit} involves a double randomness,
and $M_T$ is a measurable function of $\xi$. 
We approximate $M_T$ by 
\begin{equation}
    \Obs{M}(n)
\eqdef
\sqrt{2T}\begin{bmatrix}
    \upsigma^2_-(n)/\sqrt{\Obs{Q}_-(n)} & 0 \\
    0 & \upsigma^2_+(n)/\sqrt{\Obs{Q}_+(n)}
    \end{bmatrix}.
\end{equation}
As the Gaussian vector $\mathbf{G}$ is isotropic, 
we define for a level of confidence~$\alpha$ the 
quantity $q_\alpha$ by $\PP[|\mathbf{G}|\leq q_\alpha]=1-\alpha$. 
This quantity is easily computed
since $|\mathbf{G}|^2$ follows a $\chi^2$ distribution with two degrees of freedom.
Our \emph{confidence region} of level $\alpha$ is 
the ellipsis
\begin{equation}
    \mathcal{R}_\alpha=\Set*{
	\begin{bmatrix}
	    \upsigma_-(n)^2\\
	    \upsigma_+(n)^2
	\end{bmatrix}
	+\frac{q_\alpha}{\sqrt{n}}\Obs{M}(n)
	\begin{bmatrix}\cos(\theta)\\ \sin(\theta)\end{bmatrix}
	\given \theta\in [0,2\pi)}.
\end{equation}

Our \textbf{rule of decision} is then: reject the Null Hypothesis $(H_0)$ if the diagonal line $s:[0,+\infty)\mapsto (s,s)$
does not cross $\mathcal{R}_\alpha$.

\subsection{Numerical simulations}

We perform numerical simulations
to check the reliability of the estimation of $(\sigma_-,\sigma_+,r)$
as well the hypothesis test \textquote{$\sigma_+=\sigma_-$}.
For this, we simulate $N=\num{1000}$
paths with daily data over five years  for the three sets of
parameters given in Table~\ref{table:simu:set}.
The density of the estimated values of $\sigma_-$, $\sigma_+$ and $r$
are shown in Figure~\ref{fig:simu:sm:sp}. The proportion of rejection
of the hypothesis test \textquote{$\sigma_+=\sigma_-$} are given 
in Table~\ref{table:simu:hyp}. A good agreement is then observed
for the parameters $\sigma_-$, $\sigma_+$ and the threshold. 

The estimation of $\mu_+$ and $\mu_-$, not shown here, presents a large variance,  as
expected, since when $\mu_+=\mu_-=0$ the process is only null recurrent. This
is discussed in full details in \citet{lejay_pigato2}.

\begin{table}
    \centering
\begin{tabular}{r l c c c c c}
    \toprule
    & & $\sigma_-$ & $\sigma_+$ & $\mu_-$ & $\mu_+$ & $m$ \\
    \midrule
    set 1 & $\sigma_+\ll \sigma_-$ & \SI{80}{\%\per\year} & \SI{30}{\%\per\year} & \num{0} & \num{0} & $1$\\
    set 2 & $\sigma_+\approx \sigma_-$ & \SI{50}{\%\per\year} & \SI{30}{\%\per\year} & \num{0} & \num{0} & $1$\\
    set 3 & $\sigma_+=\sigma_-$ & \SI{30}{\%\per\year} & \SI{30}{\%\per\year} & \num{0} & \num{0} & $1$\\
    \midrule
    \multicolumn{2}{c}{} & 
    \multicolumn{5}{c}{$S_0=\SI{1}{\$} $}\\
    \bottomrule
\end{tabular}
\caption{\label{table:simu:set} Set of yearly parameters used for simulations.}
\end{table}

\begin{figure}[ht!]
    \begin{center}
	\includegraphics{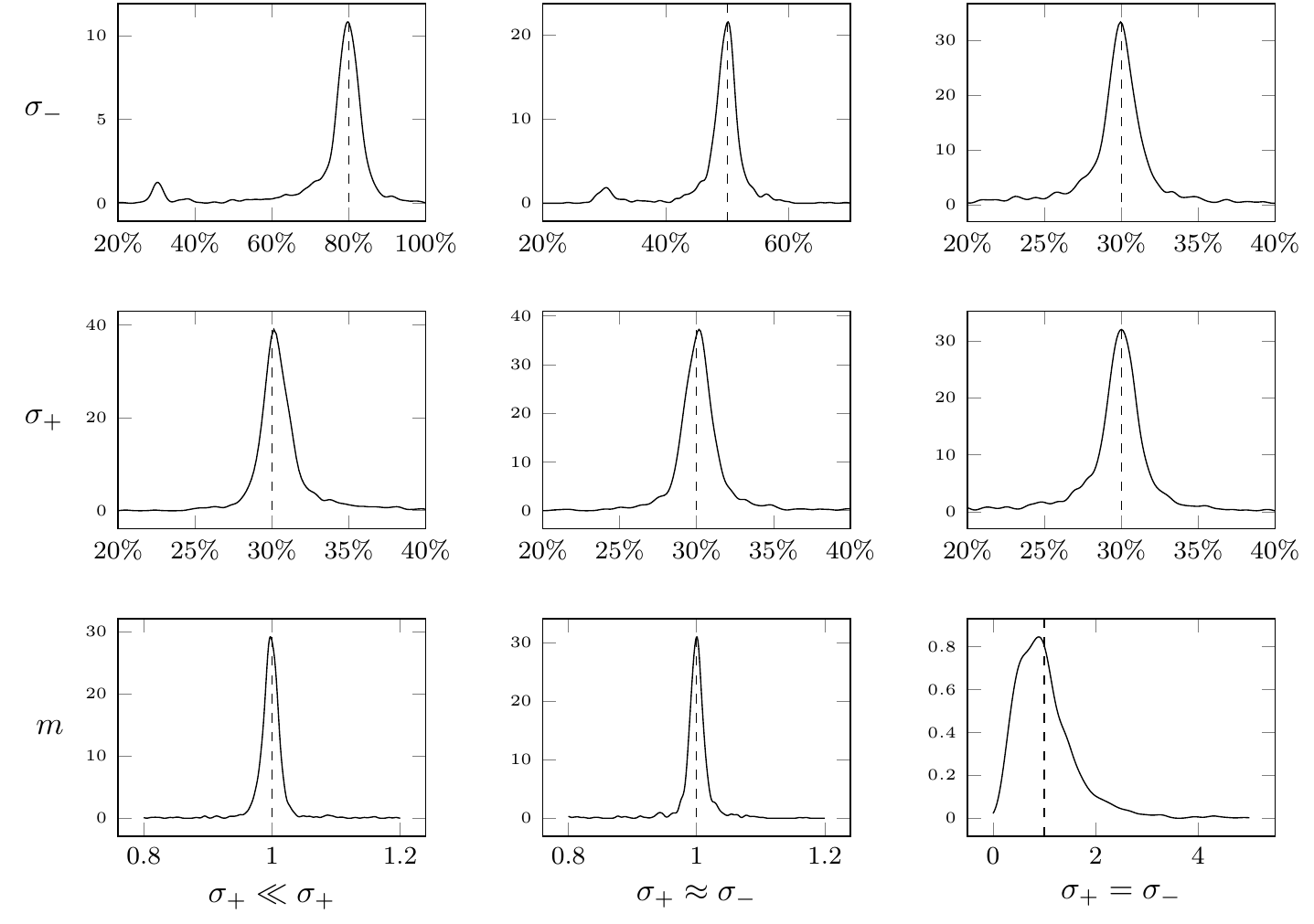}
	\caption{\label{fig:simu:sm:sp} 
	    Density of the estimated values of $\sigma_-$, $\sigma_+$
	    (yearly) and the threshold $m$ for the three sets of parameters
	    of Table~\ref{table:simu:set}, 
	    with \num{1000} simulations per set.
    }
    \end{center}
\end{figure}

\begin{table}
    \centering
    \begin{tabular}{r ccc}
	\toprule
	$\sigma_+\ll\sigma_-$ & $\sigma_+\approx\sigma_-$ & $\sigma_+=\sigma_-$\\
	\midrule
	\si{81}{\%} & \si{81}{\%} & \si{14}{\%} \\
	\bottomrule
    \end{tabular}
\caption{\label{table:simu:hyp} Proportion of rejection 
of the null hypothesis $(H_0)$ \textquote{$\sigma_+=\sigma_-$}
with a \si{95}{\%} confidence level
for the three sets of parameters of Table~\ref{table:simu:set},
with \num{1000} simulations per set.}
\end{table}

\subsection{Empirical result on the 2005 -- 2009 data}

As the drift is small, it should not affect this test. Therefore we assume 
through all this section that $b_-=b_+=0$.

In Figure~\ref{fig:confidenceregion}, we apply this rule to our data. 
The null hypothesis $(H_0)$ \textquote{$\sigma_-=\sigma_+$}
is rejected for all the stocks except for
PCG, meaning that $\sigma_-\not=\sigma_+$ should be considered for twenty out of twenty-one stocks. 
The normalized occupation time $O_+$ for PCG 
is close to \SI{99}{\%}. This may explain the elongated shape of the associated
confidence region.

\begin{figure}[ht!]
    \begin{center}
	\includegraphics{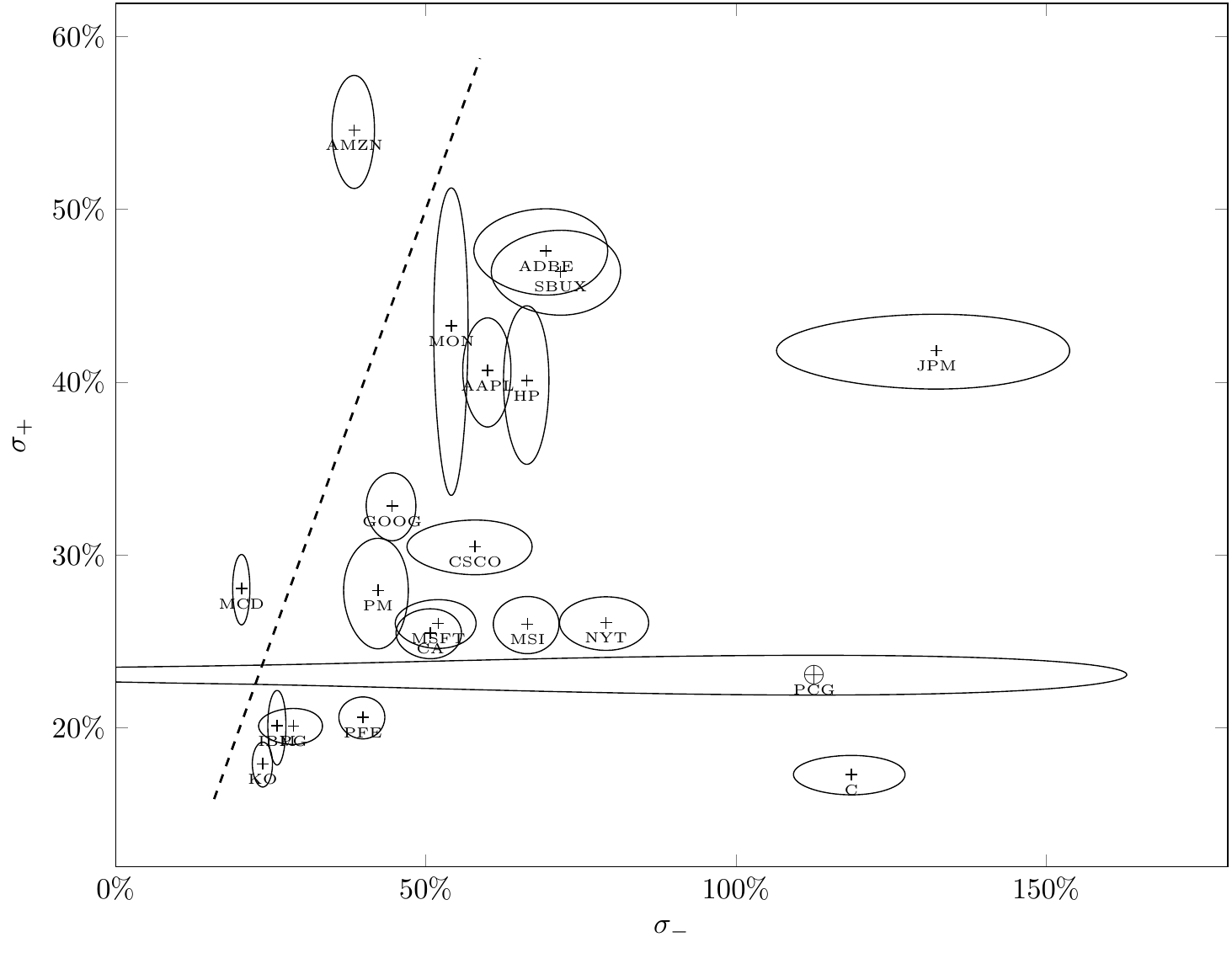}
	\caption{\label{fig:confidenceregion} 
	    Confidence regions for $(\sigma_-,\sigma_+)$:
	Each point is the value of the stock in the $(\sigma_-,\sigma_+)$-plane. 
	Confidence regions at $\SI{95}{\%}$ are the ellipsis in the $(\sigma_-,\sigma_+)$-plane 
	around the points. Points marked by $\oplus$ are the ones
	for which the Hypothesis $\sigma_-=\sigma_+$ is not rejected. 
    Points marked by $+$ are the ones for which this Hypothesis is rejected.}
    \end{center}
\end{figure}

In Figure~\ref{fig:loglik}, we plot the approximated log-likelihood
$\LogLik(i)$ against $r^{(i)}=\log m^{(i)}$ for three stocks. We see
that $\LogLik(i)$ may have one main peak (for CSCO), two main peaks 
(for GOOG) or be ``flat'' as for PCG. 
A steep peak means that it is clear where the threshold level should be taken,
and the procedure is more stable. In these cases, $\sigma_-$ is likely to
differ from $\sigma_+$ and leverage effect occurs.

\begin{figure}[ht!]
    \begin{center}
	\includegraphics{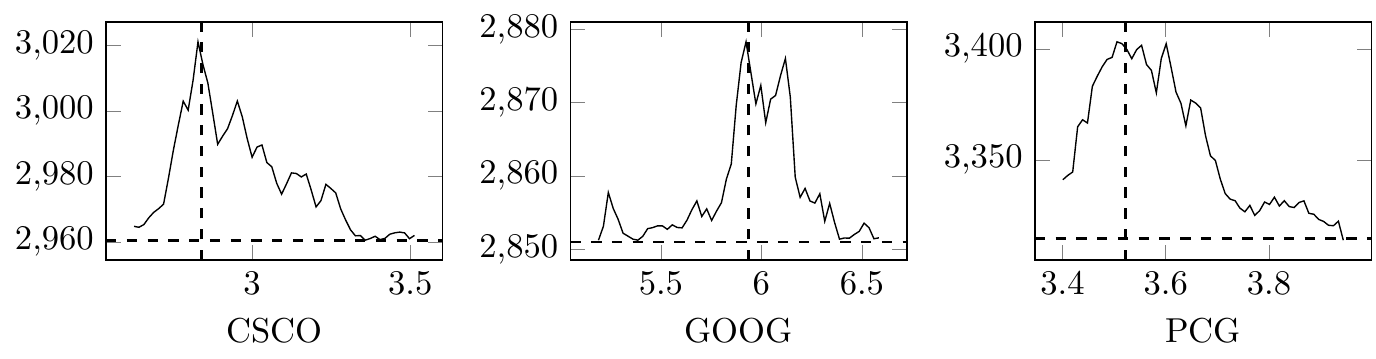}
	\caption{\label{fig:loglik} 
	    The (approximated) log-likelihood $\LogLik(i)$ given by~\eqref{eq:loglik}
	    against the possible threshold $r^{(i)}=\log m^{(i)}$ for the stocks 
	    CSCO, GOOG and PCG. The vertical dashed line represents 
	    the threshold $r^{(i)}$ which maximizes~$\LogLik{(i)}$, hence the estimated
	    $r$. The horizontal dashed line represents the value of the log-likelihood
	    of the drifted Brownian motion (that is $\sigma_-=\sigma_+$ and $b_-=b_+$).
    }
    \end{center}
\end{figure}

\subsection{Comparison with a non-parametric estimator}

Non-parametric estimation assumes nothing on
the underlying volatility and drift coefficients~\parencite{kutoyants2004,iacus}. 
The Nadaraya-Watson estimator provides us with such an estimator~\parencite{iacus}.
We then compare graphically our
estimations with the non-parametric estimation of the coefficients of the log-price. 
For this, we use the \texttt{R} package \texttt{sde}~\parencite{iacus}. 
In Figure~\ref{fig:nonparametric}, we present the results for the three stocks already used in Figure~\ref{fig:loglik}. 
More figures may be found in \citet{LP2}.
Most of the stocks seem to exhibit a behavior similar to the one presented
here, with a sharp variation of both the volatility and the drift. Again, 
this reinforces the idea that regime switching holds for most of the stocks.

\begin{figure}[ht!]
    \begin{center}
	\includegraphics{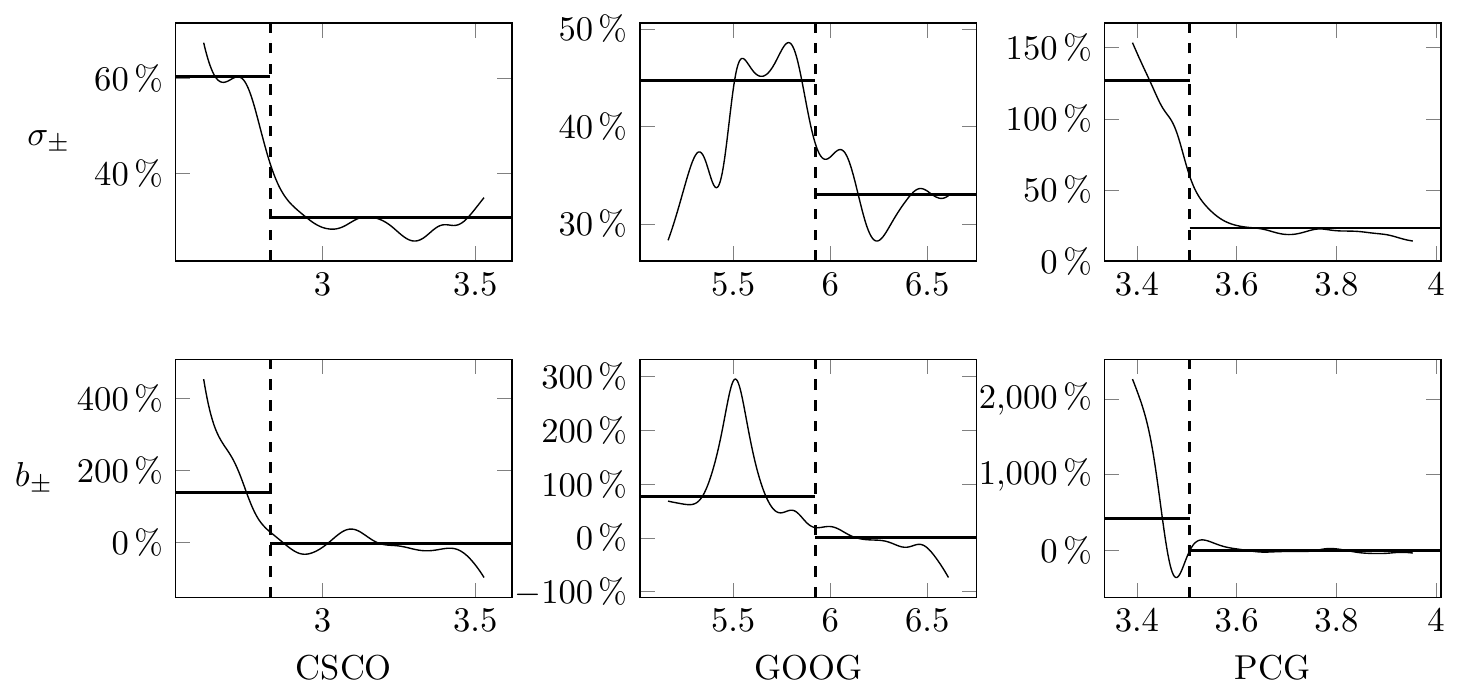}
	\caption{\label{fig:nonparametric} 
	    Non-parametric estimation of $\sigma$ and $b$ for the log-price
	    of the stocks CSCO, GOOG and PCG with a Nadaraya-Watson estimator. 
	    The vertical dashed line represents the 
	    choice of the threshold. The horizontal lines represent
	    the estimated values of $(\sigma_-,\sigma_+)$ (top) and $(b_-,b_+)$ (bottom). 
    }
    \end{center}
\end{figure}


\section{Leverage and mean-reversion effects for the S\&P 500 stocks}\label{sec:SP500}

We apply our estimators to the S\&P 500 stock prices over 
the three periods of five years each:  1/1/2003 -- 31/12/2007, 1/1/2008 -- 31/12/2012 and 
1/1/2013 -- 31/12/2017. The second period contains the 2008 financial crisis. 

Estimators are thus computed on the three periods for $332$ 
of the stocks out of the $500$, failures being due to the lack of complete data 
or to stock prices with at most \SI{5}{\%} of the observations 
on each side of the estimated threshold. The latter exclusion aims at avoiding outliers.

We plot the estimated values of $(\sigma_-,\sigma_+)$ and $(b_-,b_+)$ in Figure~\ref{fig:sp500}. 

\begin{figure}[ht!]
    \begin{center}
	\includegraphics{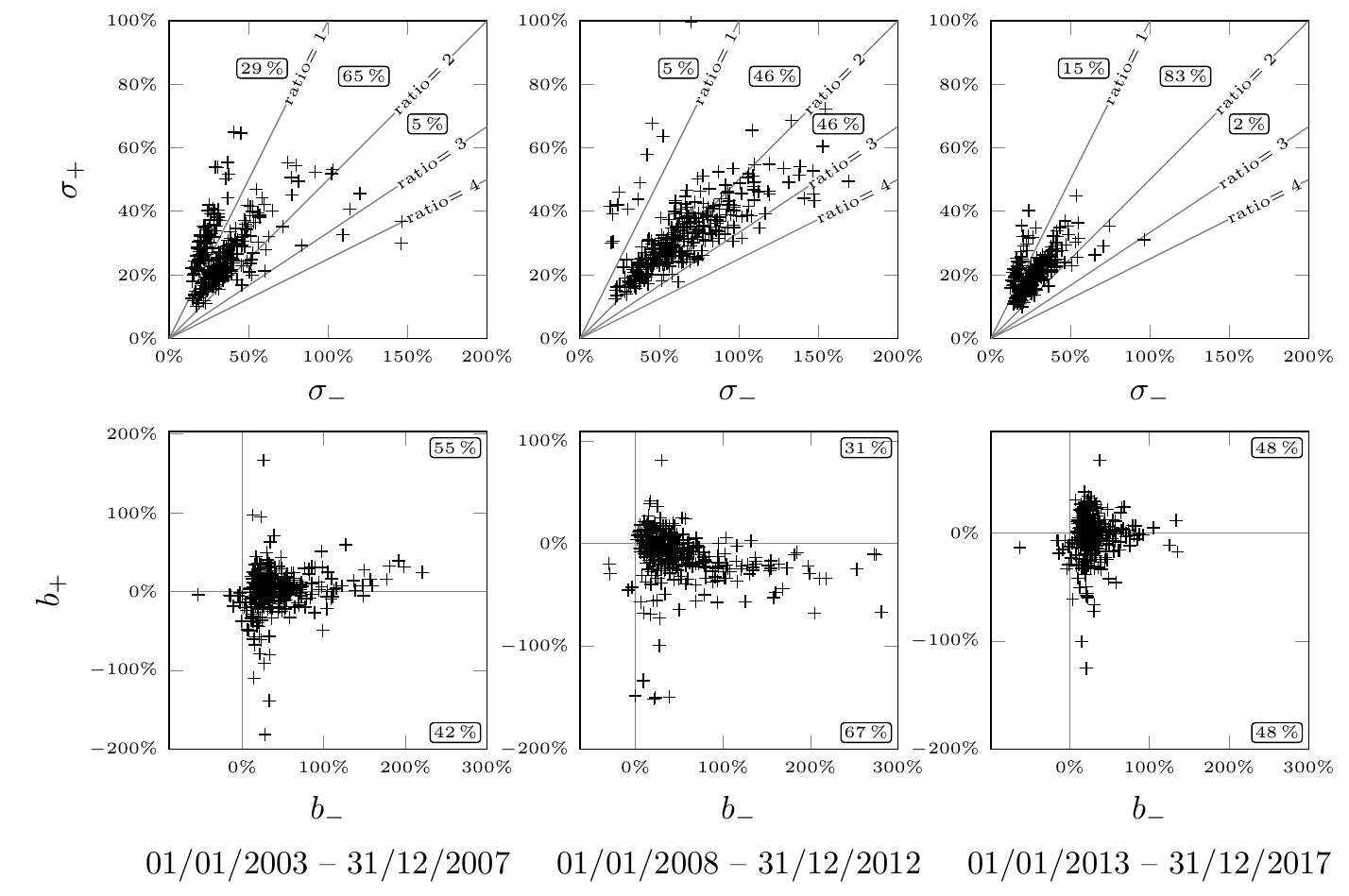}
	\caption{\label{fig:sp500} 
	    Yearly estimated values of $(\sigma_-,\sigma_+)$ and $(b_-,b_+)$ 
	    on the daily close prices 
	    for the stocks in the S\&P 500. The percentages
	    in boxes indicates the proportions of stocks in each region
	    delimited by the solid gray lines.
	}
    \end{center}
\end{figure}

In the 2008-2012 period, containing the 2008 financial crisis, 
the stock prices exhibit higher ratio of $\sigma_-/\sigma_+$, 
hence stronger leverage effects. 
For the 2008-2012 period, the
hypothesis test \textquote{$\sigma_-=\sigma_+$} is rejected for all the stocks.
For the 2003-2007 period (resp. 2013-2017), the test is rejected
for all but five (resp. six) stocks.

These findings seem to be consistent with \citet{ait-sahalia13a}, which shows leverage effect in an aggregated form
through the correlation between the VIX (involving the ex-ante volatility)
and the log-return of the indices of the S\&P 500 for the period
1/1/2004 -- 12/12/22007.

Let us now look at the  drift coefficient. First, we notice that $b_-$ is almost always positive, but $b_+$ can be positive or negative. 
For the 2008-2012 period, which contains the 2008 financial crisis, $224$ stocks show a mean-reverting
behavior ($b_->0,b_+<0$) against $140$ (resp. $161$) for the 2003-2007 period
(resp. the 2013-2017).
A possible interpretation of such results is that, on periods not involving
financial crisis, prices tend to increase in time and therefore also the
estimated value of $b_+$ is positive relatively often. When a crisis occurs,
prices oscillate more, going down as well as up, and thus giving in most cases
estimated values of $b_+ < 0$, $b_- >0$. This seems to be consistent with the
results of Section \ref{sec:empirical} and \citet{mota14a}.
 
Let us also report here a similar finding of \citet{SPIERDIJK2012228}, together with 
one of its economic interpretations given in the same paper: ``Our findings suggest that expected returns diverge away from their long-term value and converge back to this level relatively quickly during periods of high economic uncertainty; much faster than in more tranquil periods.
When the economic uncertainty dissolves, expected returns are likely to show a substantial increase
in value during a relatively short time period, which could account for such high mean-reversion
speed. Measures and interventions by financial and government institutions to restore financial
stability may also speed up the adjustment process.''

In Figure~\ref{fig:sp500:threshold}, we plot a normalized log-threshold
by dividing the log-threshold $r$ by the mean log-price (in order to get a normalized value)
for one period against the next one. We observe that this ratio ranges
between $0.8$ and $1.2$ for each stock and each period. 
The normalized log-threshold is statistically higher for the 2013-2017 period
than for the 2008-2012 period.

\begin{figure}[ht!]
    \begin{center}
	\includegraphics{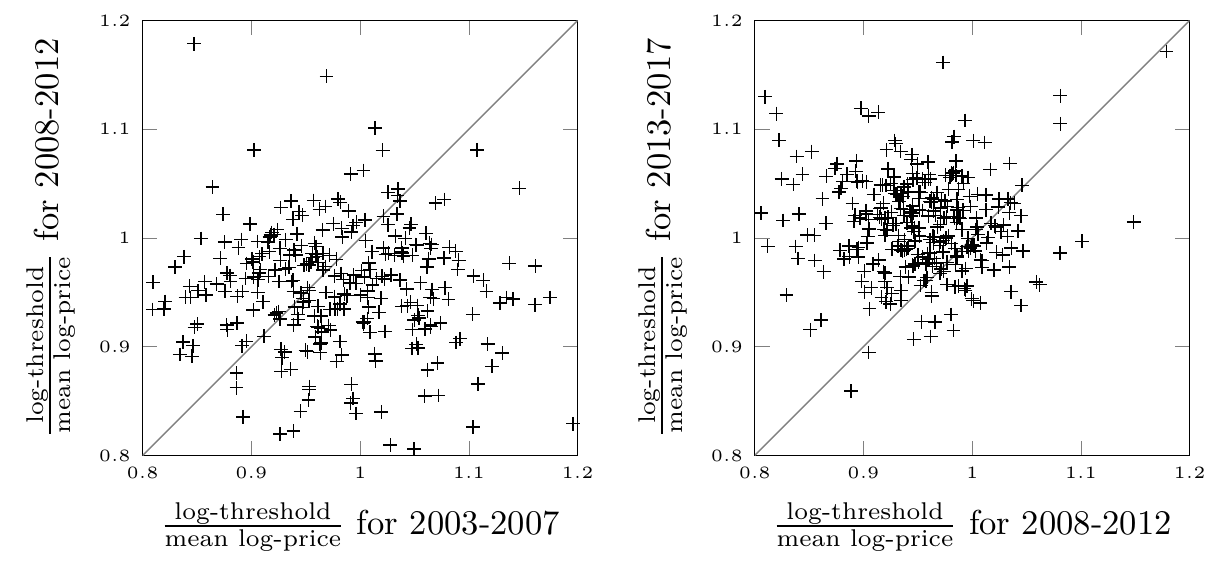}
	\caption{\label{fig:sp500:threshold} 
	    Ratio of the log-threshold $r$ over 
	    the mean log-price for two consecutive periods.
	    The solid line is $y=x$.
	}
    \end{center}
\end{figure}

These facts seem to indicate than in the crisis period, 
a stronger leverage effect is more likely to occur
at a lower threshold.

\section{Conclusion}

\label{sec:conclusion}

Leverage effects in finance have been the subject of a large literature with
many empirical evidences. 
The \emph{geometric oscillating Brownian motion} (GOBM) studied in this article mimics such leverage effect.
This model can be interpreted as a continuous-time
version of the self-exciting threshold autoregressive model (SETAR).

We have shown its validity on real data and exhibited evidence in favor of leverage
effects.  We detect a mean-reverting behavior for most of the stocks in periods of financial crisis, less so in periods not containing major events, in agreement with \citet{SPIERDIJK2012228}.
Our estimations are consistent with the ones of M.~Esquível and
P.~Mota based on least squares.

Our model is simple and does not aim at capturing other stylized facts. 
It could serve as a basic building brick for more complex models. 
Our rationale is that the GOBM is really tractable while offering more flexibility than 
the Black-Scholes model:
\begin{itemize}[nosep,partopsep=0pt,leftmargin=1cm]
    \item The estimation procedure is simple to set up. 
    \item Simulations are easily performed. 
    \item The market is complete.
    \item Option pricing could be performed through analytic or semi-analytic approach 
without relying on Monte Carlo simulations.
\end{itemize}

In addition, our model and estimation procedure could serve other purposes. In this model the leverage effect is a consequence of a spatial segmentation in which the dynamics of the price changes according to a threshold. 
The same estimation procedure could also be applied in short time windows in order to detect sharp changes, hence reflecting temporal changes, as for regime switching models involving Hidden Markov models.

Another possible application of the GOBM, and more generally of local volatilities with discontinuities, would be to introduce such features in more complex models. 
The properties we showed in the present paper and their capability of reproducing extreme skews in the implied volatility \parencite{pigato} suggest that such discontinuities could be a tractable way to introduce asymmetries and regime changes in other models \parencite{interestrate}.

\section*{Acknowledgements}

P. Pigato gratefully acknowledges financial support from ERC via Grant CoG-683164.
The authors are grateful to the anonymous reviewers for their useful suggestions.



\printbibliography

\end{document}